\documentclass[12pt,format=acmsmall,review=false, nonacm=true]{article}
\usepackage{amsthm,amssymb}

\usepackage{amsmath}
\makeatletter
\newcommand{\leqnomode}{\tagsleft@true\let\veqno\@@leqno}
\makeatother

\usepackage{booktabs} 
\usepackage[ruled]{algorithm2e} 

\SetAlFnt{\small}
\SetAlCapFnt{\small}
\SetAlCapNameFnt{\small}
\SetAlCapHSkip{0pt}
\IncMargin{-\parindent}

\usepackage{todonotes}
\usepackage{natbib}
\setcitestyle{authoryear}
\usepackage{tikz}

\usepackage[english]{babel}

\usepackage{amsmath, mathdots, latexsym, amsbsy,mathtools}  
\usepackage{graphicx}
\usepackage{dsfont}
\usepackage{enumerate}
\usepackage{bbm}

\usepackage[letterpaper, margin=1in]{geometry}
\usepackage{hyperref}

\usepackage{cleveref}

\newtheorem{lemma}{Lemma}

\newtheorem{theorem}{Theorem}
\newtheorem{claim}{Claim}
\newtheorem{example}{Example}

\newcommand{\rev}{\ensuremath{\mathsf{rev}}}

\newcommand{\cl}{\text{cl}}
\newcommand{\conv}{\text{conv}}
\newcommand{\R}{\mathbb{R}}

\usepackage{thm-restate}

\let\epsilon\varepsilon

\title{The Simplicity of Optimal Dynamic Mechanisms
\footnote{This work was partially funded by the ANID Chile through Grants FB210005, AIM230004, AFB230002, and ACT210005.}
}
\ifdefined\soda
\author{}
\else
\author{Jos\'e Correa$^*$, Andr\'es Cristi$^\dag$, Laura Vargas Koch$^\ddag$}
\fi
\begin{document} 

\ifdefined\soda
\date{ }
\else
\date{%
\small
    $^*$Department of Industrial Engineering, Universidad de Chile, Santiago, Chile\\ \url{correa@uchile.cl}\\
    $^\dag$College of Management of Technology, EPFL, Lausanne, Switzerland\\ \url{andres.cristi@epfl.ch}\\
    $^\ddag$Department of Mathematics, University of Bonn, Bonn, Germany\\ \url{vargas-koch@dm.uni-bonn.de}\\[3ex]
}

\fi


\maketitle


\begin{abstract}

A fundamental economic question is that of designing revenue-maximizing mechanisms in dynamic environments. This paper considers a simple yet compelling market model to tackle this question, where forward-looking buyers arrive at the market over discrete time periods, and a monopolistic seller is endowed with a limited supply of a single good. In the case of i.i.d. and regular valuations for the buyers, \citet{BS16} characterized the optimal mechanism and proved the optimality of posted prices in the continuous-time limit. Our main result considers the limit case of a continuum of buyers, establishing that for arbitrary independent buyers' valuations, posted prices and capacity rationing can implement the optimal anonymous mechanism. Our result departs from the literature in three ways: It does not make any regularity assumptions, it considers the case of general, not necessarily i.i.d., arrivals, and finally, not only posted prices but also capacity rationing takes part in the optimal mechanism. Additionally, if  supply is unlimited, we show that the rationing effect vanishes, and the optimal mechanism can be implemented using posted prices only, à la \citet{B08}.

\end{abstract}
\newpage
\section{Introduction}

Revenue management has emerged as one of the most important business practices in the last four decades. One of its primary concerns is the appropriate use of dynamic pricing and capacity allocation to improve profitability~\cite{TV04,P21}. 
Dynamic pricing is the practice of varying the price of a product to reflect varying market conditions, while capacity allocation or rationing is the practice of deliberately understocking products to induce high-value customers to purchase at higher prices. Both these techniques are ubiquitous in essentially all industries using revenue management tools.

Consider markets where a company sells multiple (almost) identical goods over several time periods to buyers that arrive over time. The degree to which sellers use dynamic pricing and capacity allocation varies from market to market. For example, airlines use these strategies extensively. They discriminate depending on the sales channel, they create several tiers with different prices and availabilities, and prices heavily fluctuate. In retail, we also see dynamic pricing and rationing: the stock and price of a type of coat varies during the season, even in the same store. On the other side of the spectrum, we have digital goods markets. Prices of e-books and video games vary over time, but stock is rarely limited.

In principle, however, sellers could apply even more complicated schemes to increase revenue. For example, they could offer several lotteries in each period or run an optimal auction in every period.
Why do sellers use simpler mechanisms in many dynamic markets? 
Is this just because these are easy-to-implement and easy-to-understand mechanisms, or is there another reason beyond that? \citet{BS16} showed that, in fact, dynamic pricing is optimal among all possible mechanisms in the continuous-time limit of a market where buyers have i.i.d. valuations for the good and arrive according to a Poisson process. The answer of Board and Skrzypacz is then that, under certain conditions, these simpler strategies are optimal. Similarly, \citet{DMSW23} show that dynamic pricing is optimal in a market with a single (atomic) buyer. 

However, many markets, particularly large ones, fall outside the scope of these results, and we still observe simple mechanisms. This motivates our main question: In what other cases are rationing and dynamic pricing, or even just dynamic pricing actually the best options a seller has?  

\subsection{Our Results}

In this paper, we prove that in large multi-period markets, if the seller cannot discriminate based on buyers' arrival times, a simple mechanism that combines dynamic pricing and rationing maximizes the seller's revenue. And perhaps more surprisingly, we show that in the case of unlimited inventory, dynamic pricing alone is sufficient to achieve optimal revenue, even if demand varies over time.

More concretely, we consider the problem of a monopolistic seller that sells identical copies of a good over several discrete periods and before a horizon $T$. We allow incoming demand to vary over time: in each period, a new generation of buyers enters the market and it is drawn from a distribution that depends on the period. Buyers are fully rational and stay in the market until they receive a copy of the good. We assume the distributions of the buyers are publicly known, but the types are private, and that the seller commits to a mechanism from the start. We call this a \emph{multi-period market}. In the same model, \citet{B08} considered the case that the seller commits specifically to a dynamic posted-price mechanism and studied the question of how and under what conditions it is possible to compute the optimal sequence of prices.

We make two crucial assumptions. The first is that we consider the limit case of a continuum of buyers. Buyers are \emph{non-atomic} (infinitesimally small), and therefore, a deviation from the equilibrium by a single buyer cannot be perceived by the seller. The second is that the seller cannot observe the identities of the buyers and thus cannot discriminate between buyers from different generations. Specifically, we assume that the allocation and payments in any given period cannot depend on the actions of the buyers in previous periods (otherwise, the seller could use this to discriminate generations, e.g., by selling coupons to buyers that arrive early). We call mechanisms that satisfy these conditions \emph{anonymous}.

Our main result is twofold. 
First, we show that a revenue-optimal anonymous mechanism in the multi-period market has a very simple structure. For every period the seller offers two different options to buy the item. Option 1 is to get the good at a high price. Option 2 is to participate in a lottery, where the buyer gets the good only with a certain probability and pays a lower price.

\begin{theorem}
\label{thm:rationing_optimal}
    In every multi-period market, there is a mechanism achieving the optimal revenue among all anonymous mechanisms, where for every period $t \in \{1,\dots, T\}$, the seller first offers the good for some price $p_t$ and then offers a quantity $c_t$ of the good in a lottery for some price $\bar p_t < p_t$.
\end{theorem}

Second, we show that lotteries are only necessary when there is limited inventory. In other words, with unbounded supply, as is the case with digital goods markets, dynamic posted prices are an optimal mechanism in every instance.

\begin{theorem}
\label{thm:pricing_optimal}
    In every multi-period market, if the inventory constraint is non-binding, then there is a dynamic posted price mechanism achieving the optimal revenue among all anonymous mechanisms, i.e., there is an optimal mechanism where the seller offers the good for a price $p_t$ in each period $t\in \{1,\dots, T\}$.
\end{theorem}

These results depart from standard literature in the area in three key aspects. The first is that we do not restrict buyers' valuations to be identically distributed as \citet{BS16}. The second is that we do not make any sort of regularity assumptions on the distributions of valuations. In fact, we do not even require the distributions to be absolutely continuous. The third aspect is that we allow for very general discounting: the seller can have a varying discount rate on cash flows, and the buyers can have a discount factor on their valuations to model impatience and a possibly different discount rate on the payments.

An interesting implication of our results combined with the work of \citet{B08}, is that the optimal dynamic mechanism can actually be computed. \citet{B08} provides an algorithm to compute optimal dynamic prices under a mild convexity assumption on the revenue curve. We show that dynamic prices are indeed the optimal mechanism when we have unlimited inventory, so Board's algorithm finds the optimal anonymous mechanism when both conditions are satisfied.

It is easy to check that our results do not hold if we remove any of our two crucial assumptions. 
If we allow for \emph{non-anonymous} mechanisms, the optimal mechanism loses its dynamic aspect, since buyers are perfectly discriminated according to their arrival time. 
In the case of unlimited supply, the optimal mechanism for the seller is to treat each generation as a separate problem and sell at the monopoly price for each generation. More generally, in case of limited supply, the optimal non-anonymous mechanism takes the following form: the seller guesses (pre-computes) how many units to sell to each generation and then, in each period she runs an optimal static auction considering only those buyers that arrived in that period and supply equal to the pre-computed quantity.
On the other hand, if we require the mechanism to be anonymous but we have \emph{atomic} buyers, the mechanism can anonymously punish deviations. We illustrate the latter with the following example. 

\begin{example}
    Consider a market with two periods and infinite supply. In the first period, $1$ unit of non-atomic buyers arrives, all with a valuation of $1$. In the second period, $1$ unit of non-atomic buyers arrives, all with a valuation of $1/2$. A non-anonymous mechanism can set the price for first-generation buyers to be $1$ in both periods and the price for second-generation buyers to be $1/2$, extracting a total revenue of $3/2$. In contrast, no dynamic posted price mechanism can get a revenue higher than $1$. On the other hand, if we require the mechanism to be anonymous, but we have discrete buyers, we can consider an almost identical instance. Take the same values, but instead of a unit of buyers in each generation, assume each generation is composed of exactly one buyer. In the second period, the seller can ask what are the values of the two buyers present in the market. If the answers are $1/2$ and $1$, then both get the item at prices $1/2$ and $1$, respectively. If any of the two buyers deviates, no one receives the item. Again, the revenue is $3/2$, whereas with posted prices, the maximum possible revenue is $1$.
\end{example}

The same instance can be used to show that with unlimited supply, even if we further restrict the mechanism to use only prices and rationing, \Cref{thm:pricing_optimal} does not hold if we have discrete buyers. 
However, in case of a single buyer, \citet{DMSW23} showed that dynamic pricing is optimal among the class of non-adaptive lottery mechanisms (i.e., the seller decides on a set of lotteries at time 0 and cannot adapt the offered lotteries based on observed participation). The following example shows that we cannot hope for the same type of result for two or more buyers, which justifies our restriction to the non-atomic setting.

\begin{example}
\label{ex:rationing_usefull_discrete}
  Consider two buyers. The first buyer arrives at time one with valuation $v_1=1$, and the second buyer arrives at time 2 with valuation $v_2=1/2$. If the seller only uses posted prices, the maximum possible revenue is again $1$. Now, suppose the seller can additionally use capacity rationing, and consider the following strategy $(p_1,c_1)=(3/4,\infty)$ and $(p_2,c_2)=(1/2,1)$, where $c_i$ denotes the number of copies of the good that will be offered in period $i$ for a price of $p_i$. In this case, buyer two will try to purchase the item at a price of $1/2$. Given this behavior, buyer one must decide either to purchase upon arrival for a utility of $1/4$ or to wait until time $2$ and try his chances at the lottery. In this latter case, his expected utility is $(1/2)\cdot (1/2)=1/4$, (note that in this case the demand is $2$ while the available capacity is $c_2=1$). Thus, buyer one purchases in the first period. Therefore, the seller's revenue is $5/4$ and thus strictly larger than in the case of posted prices.
\end{example}

Interestingly, we can observe how \Cref{thm:pricing_optimal} emerges in this example, as the market grows and the relative size of each buyer tends to $0$. Imagine instead of a single buyer per period, we have $M$ identical buyers per period of size $1/M$. If we try to mimic the strategy from \Cref{ex:rationing_usefull_discrete}, we get that we should set $(p_1,c_1)=(1/2+\frac{1}{2(M+1)},\infty)$ and $(p_2,c_2)=(1/2,1)$, as a single buyer that deviates from period $1$ to period $2$ gets a utility of $(1/2)\cdot (M/(M+1))$. Therefore, the resulting revenue is $1+\frac{1}{2(M+1)}$, which converges to $1$, the optimal revenue for posted prices, when $M$ goes to infinity.

This difference between small and large markets could help to explain why in larger markets, like retail and digital goods markets, mechanisms are much simpler than in smaller markets like plane tickets. Since flights are not substitutes for each other, each flight can be thought of as a separate small market, where the assumption that buyers are infinitesimally small is unrealistic.

The observation that with unlimited supply the extra revenue achieved by adding capacity rationing on top of dynamic prices vanishes as the market grows large is not new. Indeed, \citet{LR08} proved this in a two-period setting, when buyers' valuations for the item take a particular shape (they are uniformly distributed). Furthermore, \citet{ZC08} extended the result of Liu and van Ryzin to allow for arbitrary distributions. We significantly extend these results by showing that dynamic pricing is optimal beyond the two-period case and among a much richer class of mechanisms. 

It is worth noting that in the multi-period case, even if we restrict our attention to rationing plus pricing, the analysis of both the equilibrium and the optimization problem is significantly more involved than in the two-period case.
In two periods, a buyer faces essentially a single decision: whether to buy in the first period. It is possible, then, to pin down the behavior of the buyers with an explicit formula. In the case of more than two periods, both the seller and the buyers face dynamic optimization problems, and the interactions between buyers become non-trivial, which means a significantly different machinery is required for the analysis.

\begin{figure}
    \centering
    \includegraphics[width=0.24\textwidth]{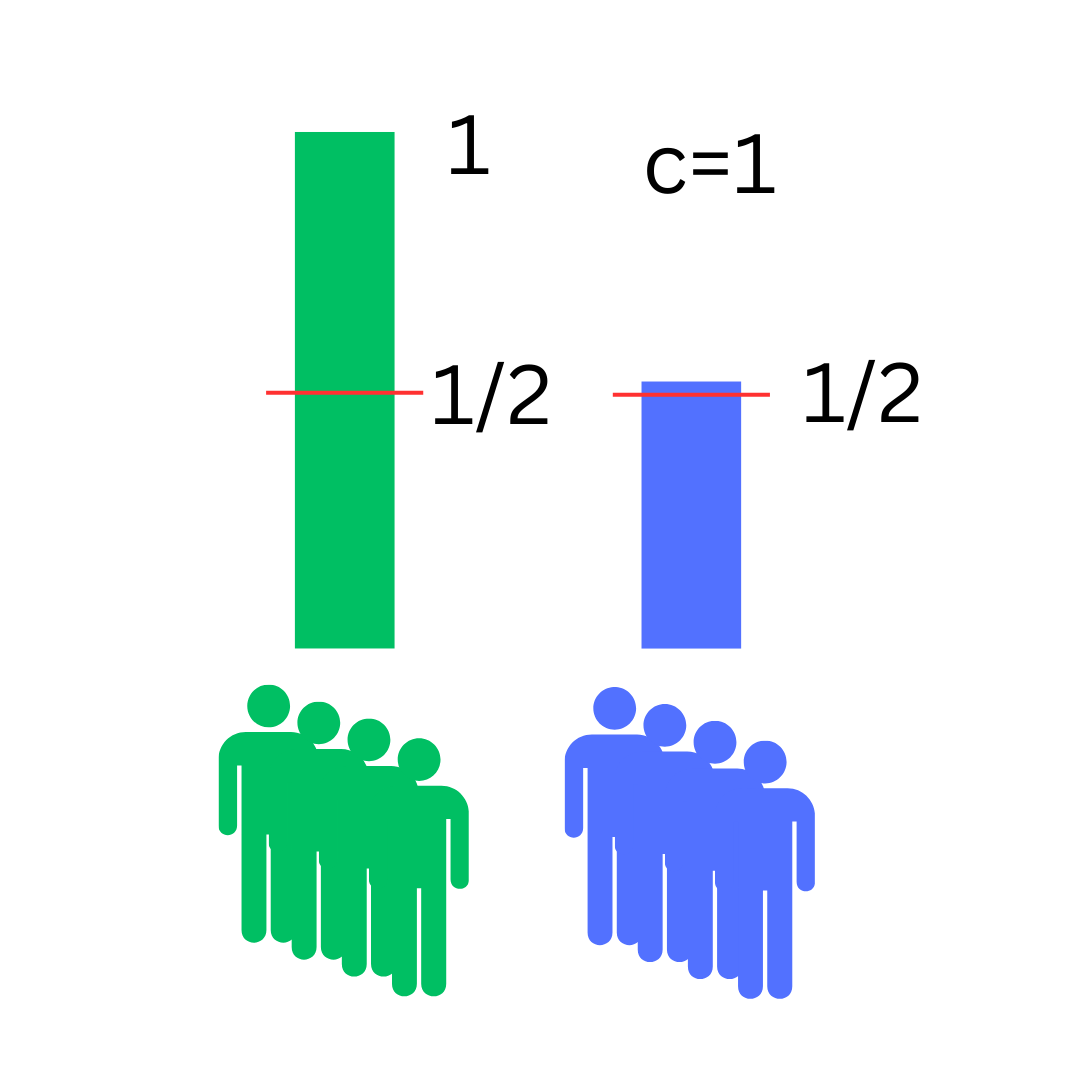}
        \includegraphics[width=0.24\textwidth]{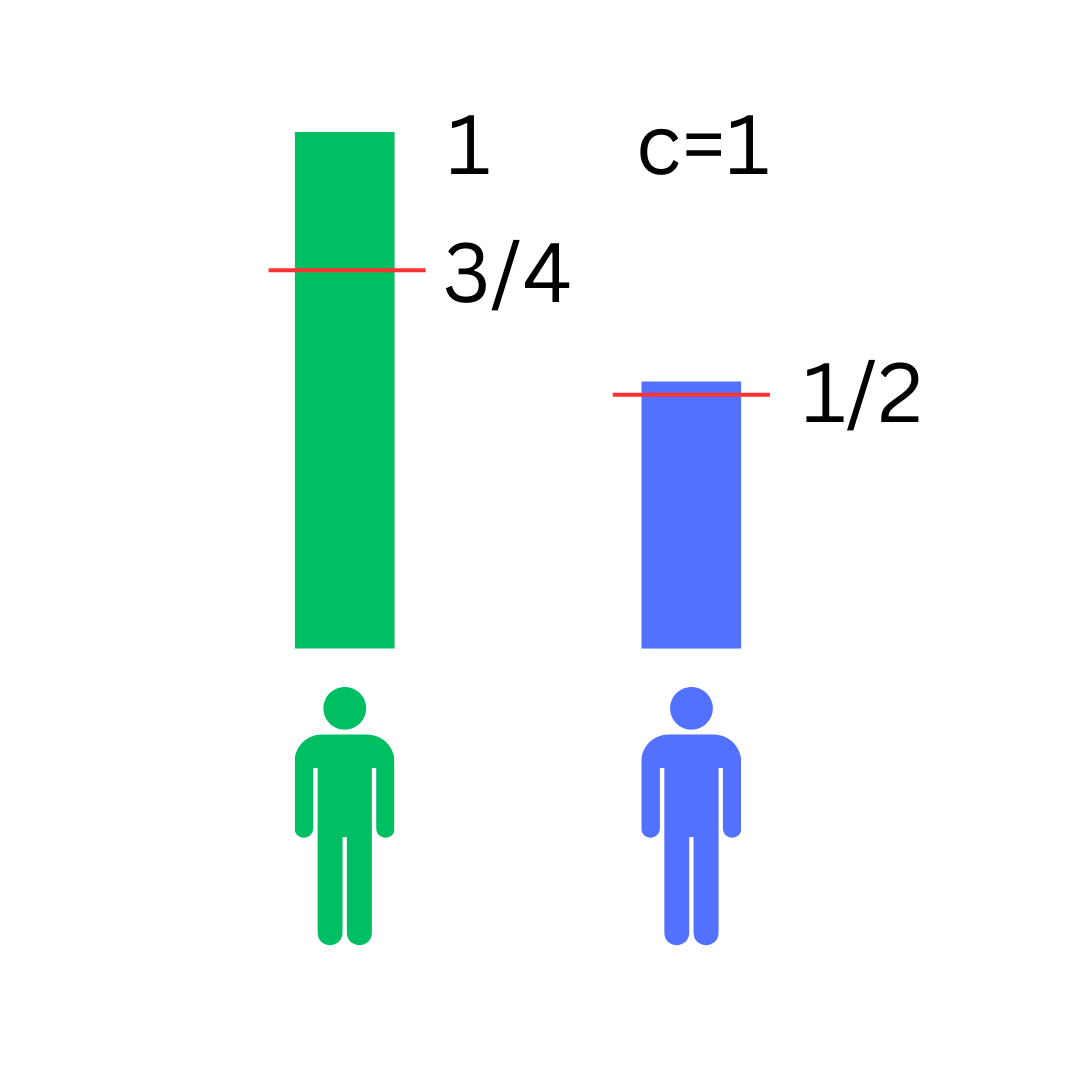}
            \includegraphics[width=0.24\textwidth]{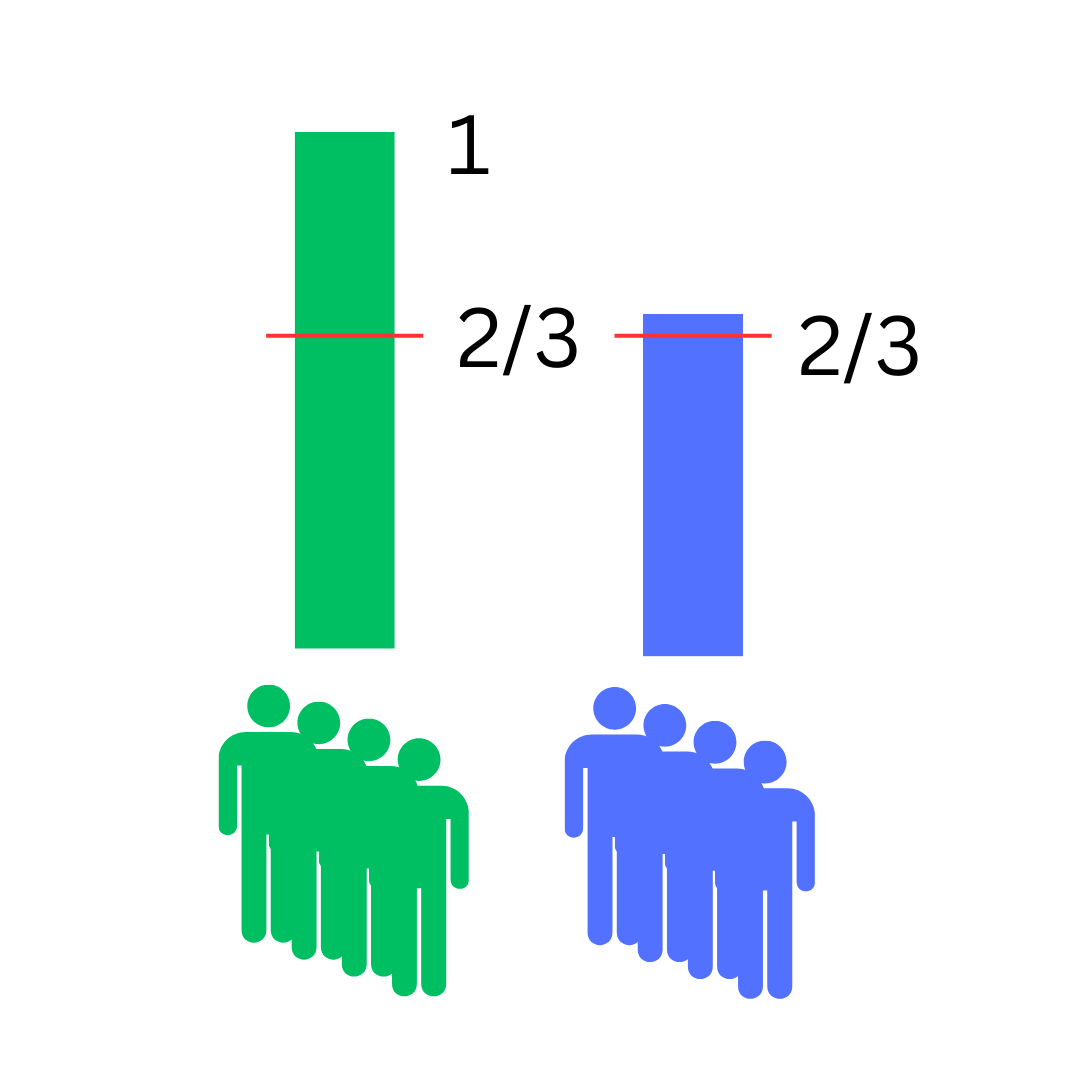}
                \includegraphics[width=0.24\textwidth]{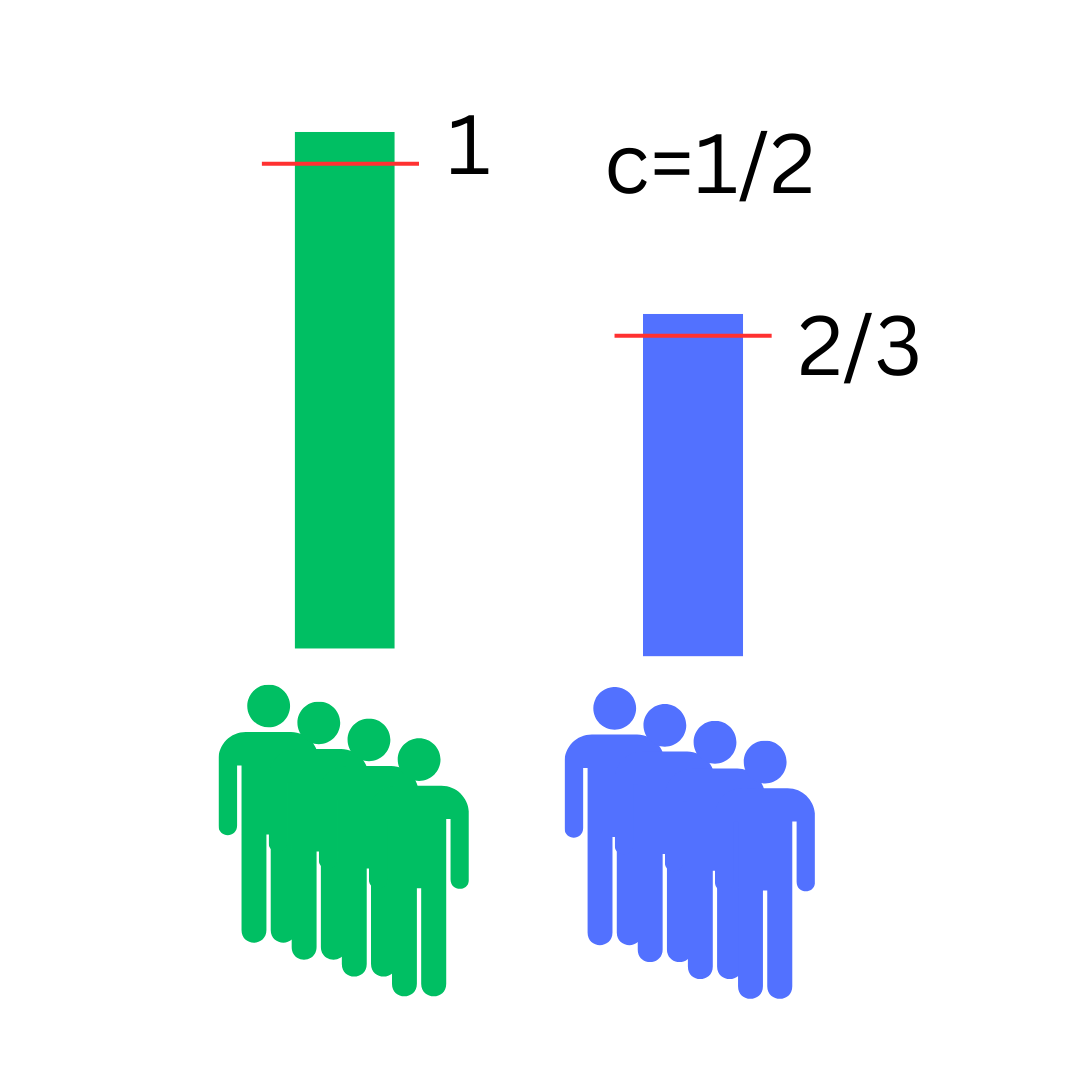}
    \caption{The pictures illustrate the examples where we describe when rationing is useful from a seller's perspective. 
    In the two pictures on the left we compare the situation in a market with discrete buyers of mass one to a market with many non-atomic buyers of the same mass. While in the large market in picture 1 rationing is of no help, this changes with atomic buyers as in picture 2, see also \Cref{ex:rationing_usefull_discrete}.
    A similar situation can be observed when considering a large market with finite inventory, here of $\frac{3}{2}$. While without rationing the prices marked in red in picture 3 are optimal, with rationing the seller can do better as depicted in picture 4, see also \Cref{ex:rationing_usefull_inventory}. 
    }
    \label{fig:enter-label}
\end{figure}


In our last example, we show that with bounded supply, lotteries are indeed necessary to achieve the optimal revenue, so there is a strict separation between the cases covered by \Cref{thm:pricing_optimal,thm:rationing_optimal}. An example in the same spirit was already presented by \citet{loertscher2020monopoly}, who studied under what conditions a seller might want to use rationing, in a single-period setting.
\begin{example}
\label{ex:rationing_usefull_inventory}
Both at time one and time two, buyers of mass one arrive. The buyers arriving at time one have a value of one while the buyers arriving at time two have a value of $\tfrac{2}{3}$. The seller has an inventory of $\tfrac{3}{2}$. 
With only prices the best the seller can do is to either offer the good for a price of one and sell one unit or to offer the good for $\tfrac{2}{3}$ and sell all $\tfrac{3}{2}$ units. In both cases, the total revenue is one. 
With rationing on the other hand, the seller can offer one unit of the good in period one for $\tfrac{5}{6}$ and half a unit of the good in period two for $\tfrac{2}{3}$. 
In this case, for a high-valued buyer, deviating from period one to period two is not advantageous since a buyer gets the good in period two only with a 50\% chance (there is only half a unit available, but buyers of mass one aim to buy). 
Thus, when using a lottery the revenue of the seller is $\tfrac{7}{6}$, which is strictly larger than in the setting with only prices, where the revenue was one.
\end{example}

The example shows that posted prices are not enough to achieve the optimal revenue. Thus, mechanisms that use dynamic pricing combined with rationing are, in some sense, the simplest mechanisms for achieving optimal revenue.

Due to \cite{M81}, it is well known that the optimal static mechanism can be described as the maximizer of the concavified revenue curve. In the case of unlimited supply, the revenue curve and its concavification coincide in the unrestricted maximizer, so the optimal mechanism takes the form of a fixed price. 
On the other hand, in the case of limited supply, the maximizer of the constrained problem might correspond to a point where the two curves do not coincide, so the concavification corresponds to a convex combination of two points in the revenue curve. Therefore, the optimal mechanism can be implemented by a high price and a low price with rationing.

\subsection{Technical Overview}

To prove our results, we first observe that for every anonymous mechanism, the outcome can be expressed by functions $r_t$ and $p_t$ for $t \in [T]$, where $r_t(v)$ expresses the probability that a buyer with valuation $v$ present at time $t$ obtains an item in this period, and $p_t(v)$ is the expected payment in this period. 
In general, one could think of anonymous mechanisms that additionally depend on publicly observable information and cannot be described by such functions. For example, if the mechanism offers two lotteries in the first period, the allocation in the second period could depend on the demand for each lottery in the first period. However, this has no real effect in the allocation in our non-atomic setting, as unilateral deviations (infinitesimal) cannot be detected from publicly observable events, such as the mass of served buyers. So, for any mechanism that uses public information from previous periods, there is an equivalent mechanism (that results in the same outcome) that does not use any information from previous periods.


By applying the direct-revelation principle, we can assume that in the mechanism, each buyer reports their true valuation in each period, and therefore, the mechanism is fully described by the functions $r_t$ and $p_t$. We show that in an incentive-compatible mechanism, the expected utility of a buyer with a given value and present at a given period can be expressed by a formula that depends only on the allocation functions $r_t$ (see \Cref{lem:utility_formula}). We do this by recursively applying the envelope theorem to write the utility function in a period as the integral of a formula that depends on the allocation function $r_t$ and the derivative of the next period's utility. Since we allow for general discounting,\footnote{The discount factor can vary over time and can be different for the seller and for the buyers, so the seller could make a profit by exploiting this difference.} to obtain this integral formula in our dynamic environment, we need two extra constraints on the mechanism: that all payments are nonnegative, and that a buyer pays only in the period where they receive the item. This ensures that the utility at value zero is zero and that we can compute the utility via a recursive formula.

We then prove that we can restrict our attention to monotone allocation functions (see \Cref{lem:r_monotone}). Our analysis is inspired by Myerson's analysis of the static case. However, we face significant challenges in the dynamic setting. We first have to untangle the dependencies of the utility functions on the allocations (notice that the utility in a period depends on the allocation functions of every period in the future) and then exploit the fact that not only each utility function is convex but also the difference between consecutive utility functions is convex. When the discount factor (on the value of the buyer) is strictly decreasing, this argument is enough to prove the monotonicity. However, if the discount factor remains constant for some pair of consecutive periods, there can be mechanisms with non-monotone allocations. We fix this issue by showing that, in that case, we can modify the mechanism so that the revenue only increases and the resulting allocation functions are monotone (see \Cref{subsec:weakly_decreasing}).

Finally, we show that the revenue and the total mass of sold items are fully determined by the allocation functions $r_t$ and, moreover, are linear on each allocation function $r_t$. This is a non-trivial fact, as in a given period $t$, the demand depends on the allocation on previous periods, and the behavior of the buyers depends on the allocation on future periods (they take their expected future utility into account). It turns out that the expected payment of a buyer in period $t$ can be expressed as the integral of products of only future allocation functions where each appears only once, and the probability that a buyer is present in period $t$ can be expressed as products of past allocation functions, where each appears only once. This results in the cash flow being linear in all allocation functions in a given period.
Note that, since we have products of allocation functions, the revenue and the available inventory are not linear on the vector of allocation functions $(r_1,\dots,r_T)$ as a whole. Despite that, the linearity on each is enough to prove our two theorems. If we take an optimal sequence of allocation functions, we can fix all but $r_t$ and re-optimize in terms of $r_t$. This way, we obtain a linear optimization problem on $r_t$, and a corollary of Caratheodory's theorem (see \Cref{lem:kplus1steps}) implies that the optimal $r_t$ is a piece-wise constant function with a single jump if the inventory constraint is non-binding, or with two jumps if it is binding. A simple construction shows that these allocation functions can be implemented with a single price per period or with a fixed price and one lottery per period, respectively, which concludes our proof. 

It is important to notice that our results cannot be derived and should not be understood as an iterative or a recursive application of Myerson's analysis for the static case. A modification of a mechanism in a single period results not only in a change in the cash flow in the given period but might affect all periods at once. The incentives for buyers in past periods will change, and also the demand in future periods will change. What we prove is that, even though cash flows in all periods can change by changing the allocation function in a single period, \emph{all} these changes are linear.

\subsection{Further Related Literature}

Designing and analyzing revenue maximizing mechanisms is a central economic problem since at least the seminal work of 
\citet{M81}. More recently, motivated by evolving platforms, there has been a lot of effort in understanding dynamic market design and dynamic mechanism design.
Thus, there is a broad literature on the topic, leading to several different models. One important characteristic is the type of the buyers. In our setting, buyers are long-lived and fully strategic as in the works of \citet{BS16,GMS17} and \citet{PST14}. This means buyers enter the market once, then wait until the optimal moment to buy, and only leave after that.
In contrast to this, the work of \citet{M16,PV13} and \citet{PS03} assumes buyers are endowed with a release date and a deadline, an entrance time into the market, and a corresponding exit time. This means either the buyers get the item by this deadline or they leave the market.
\citet{BM09} assume buyers stay in the market until they decide to purchase. Then, no matter whether they successfully obtained a good or not, they leave the market. In other models, as that of \citet{ZC08}, a combination of long-living and myopic buyers is considered. Moreover, \citet{loertscher2020monopoly}, consider a model where buyers can resell the purchased item and \citet{LMT22} consider a model where buyers as well as sellers enter and leave the market. 

Also, the knowledge the seller has over the buyers differs between models. In our setting, the seller knows the number of periods and the distribution functions according to which buyers enter into the market, but the valuations are private information. We assume that buyers bid a valuation in each time period where they are still present. 
In other settings, as that of \citet{GMS17} the seller has only partial information and learns about the distributions of the buyers by observation. 

Beside aiming for a general revenue optimal mechanism, also the analysis of optimal pricing and optimal rationing mechanisms has been studied extensively in the Theoretical Computer Science, Operations Research and Economics communities (see e.g. \citet*{BM09,B08,CDW16,CHMS10,DMSW23,loertscher2020monopoly,nisan2007algorithmic,P21,TV04}). For our setting, it is especially relevant that \citet{B08} provides an algorithm that computes revenue optimal dynamic prices under mild convexity assumptions. While pricing mechanisms can be optimal already in some settings, in several model variants, it can be shown that capacity rationing may be helpful. These include: in the single period setting \cite{loertscher2020monopoly}, with demand uncertainty \cite{NP07}, without commitment of the seller to the announced strategy \cite{DG99}, and for disappointment-averse players \cite{LS13}.

\section{Model}

We consider a \emph{multi-period market}. The market is composed by a seller and a continuum of buyers. 
The seller is profit maximizing and offers identical copies of a good with supply or \emph{inventory} $I$ in the market. The marginal cost of producing the good is normalized to zero. Buyers arrive over time and are long-living, which means they stay in the market until they get a copy of the good. 

The market operates during $T$ discrete time periods. At each period $t \in [T]=\{1,\ldots,T\}$ a new generation of infinitesimal buyers arrives in the market. The buyers that arrive at a period $t$  have private nonnegative valuations for the good given by a distribution function $F_t$ with bounded support, where $F_t(v)$ corresponds to the mass of buyers that arrive in period $t$ with value at most $v$. By normalizing the maximum possible valuation to 1, we assume that $F_t \colon [0,1] \mapsto [0,1]$.
Without loss of generality, we assume that in every period, a total buyer’s mass of 1 enters the market.\footnote{If there is a period $t$ where a mass $m_t<1$ enters the markets we can simply set $F_t(0)=1-m_t$ and scale the valuation distribution accordingly, since buyers with valuation 0 are irrelevant.} We call these functions the \emph{distributions of buyers}. Without loss of generality, we assume there is a finite measure $\mu$ of the interval $[0,1]$ and functions $f_1,\dots,f_T$ such that for all $t \in [T]$, $v \in[0,1]$
\begin{align*}
    F_t(v)= \int_0^v f_t(u)\, d\mu(u).
\end{align*}

Thus, a multi-period market consists of $T$ periods, distributions of the buyers $(F_t)_{t \in T}$, and a monopolistic seller with inventory $I$.
The seller can choose any mechanism, but she has only limited knowledge. The seller knows the distribution functions, but the valuation of a specific buyer is private information. That is, the seller cannot discriminate across types. Moreover, we assume the environment is anonymous, in the sense that the seller can neither observe when a buyer entered the market nor a buyer's actions in previous periods. We call a mechanism satisfying this an \emph{anonymous mechanism}.

We aim at finding a revenue-optimal mechanism.\footnote{As standard in the literature we look for revenue optimal mechanisms satisfying voluntary participation.} By the revelation principle, we can assume that this mechanism is incentive-compatible, i.e., the optimal strategy for a buyer present in period $t$ is to bid his true valuation. 
Moreover, in an anonymous mechanism the payment and allocation of a buyer present in period $t$ only depend on the buyer's actions in period $t$. Thus, an \emph{anonymous mechanism} is defined by an allocation function $r_t:[0,1] \mapsto [0,1]$, and a pricing function $p_t:[0,1] \mapsto [0,1]$ for each period $t \in [T]$. The function $r_t(v)$ describes the probability that a buyer who bids value $v$ gets served at time $t$ and the function $p_t(v)$ describes his expected payment in period $t$.\footnote{More generally, an anonymous mechanism could depend on the history of publicly available information. Consider $H_t$ the publicly available history at time $t$, which consists of the sales and payments in previous periods. A mechanism that uses $H_t$ can be simulated by one that does not and obtains the same revenue. Consider a mechanism $M$ described by functions $r_t(v,H_t)$ and $p_t(v,H_t)$. Fix the history in the equilibrium path (in which we may assume buyers report truthfully), and denote it by $\bar{H}_t$. We can define a mechanism $M'$ given by $r'_t(v)=r_t(v,\bar{H}_t)$ and $p'_t(v) = p_t(v,\bar{H}_t)$. Reporting truthfully must also be an equilibrium in $M'$, as a deviation by a single (infinitesimally small) agent does not affect the history of publicly available information $H_t$.}

\paragraph{The buyers.} Buyers are long-lived, forward-looking, and fully rational. Upon entering the market, they observe all future expected prices and assignment probabilities and anticipate the decisions of other buyers. Consequently, buyers choose a strategy optimizing their own expected utility. Buyers may be impatient and discount their valuation for the item, depending on when they receive it. For a decreasing sequence of (valuation) discount factors $\delta_1 \geq \ldots \geq \delta_T$  with $\delta_t \in (0,1]$ for all $t \in [T]$, 
a buyer of type $v$ has a value of $\delta_t\cdot v$ for receiving the item in period $t$.\footnote{Note that this captures any non-decreasing discount function as e.g. an exponential discounting function.} The utility of a buyer is his valuation for the good at the time he receives it minus the paid price.\footnote{In \Cref{sec:discount} we show that our results remain valid even when both the seller and the buyers also discount the monetary transfers.}
The expected utility of a buyer of type $v$, conditional on being present at time $t$, is given by the recursive formula
\begin{align}
\tag{U}
\begin{split}
 \label{eq:utility}   U_t(v) &= \max_{\bar{v} \in [0,1]} \left\{
\delta_t \cdot v \cdot r_t(\bar{v} ) - p_t(\bar{v}) + (1 -r_t(\bar{v})) \cdot U_{t+1}(v)
\right\} \\
U_{T+1}(v) &= 0.    
\end{split}
\end{align}
The assumption that the mechanism is incentive-compatible is equivalent to the assumption that the maximum in the previous formula is attained at $\bar{v}=v$ for all $v\in [0,1]$ and all $t \in [T]$. 
We will assume \emph{voluntary participation}, which is $U_t(v)\geq 0$. This is equivalent to $p_t(0) \leq 0$ for all $t \in [T]$, and by standard normalization arguments, we can assume $p_t(0)=0$ for all $t \in [T]$. This means that whenever we speak about a mechanism we assume it satisfies voluntary participation and whenever we give an explicit construction for a mechanism and prices, we will make sure that $p_t(0)=0$ for all $t \in [T]$.

\paragraph{The seller.} The seller chooses a mechanism with the goal of maximizing her revenue, i.e. the total prices actually paid in the game. 
Thus, for an incentive-compatible mechanism $M$, the revenue can be written as
\begin{equation}
\label{eq:revenue_general}
 \rev(M)=\sum_{t=1}^T \int_0^1 p_t(v) \cdot f^*_t(v)\, d\mu (v),    \tag{rev}
\end{equation}
where $f^*_t(v)$ denotes the density of buyers with valuation $v$ that are present at time $t$. This can be computed recursively from
\begin{align}
\tag{f*}
\begin{split}
\label{eq:fstar}
    f_1^*(v) &=  f_1(v), \\
    f_t^*(v) &=  f_t(v) + f^*_{t-1}(v) (1-r_{t-1}(v)).   
\end{split}
\end{align}
We also denote $F^*_t(v)=\int_0^v f^*_t(u)\,d\mu(u)$.

\paragraph{The mechanism design problem.} 
With the previous definitions, we can explicitly write the mechanism design problem faced by the seller, who aims to sell an inventory $I$.
\begin{equation}
\leqnomode
 \!\!\!  \!\!\!  \!\!\!   \qquad \max_{r_t, p_t, t\in[T]}\,  \rev(M) \tag{RM} \label{eq:RM}
\end{equation}
\vspace{-3ex}
\begin{align}
 \; s.t. \, & (f^*), (U), &&  \notag\\
&\sum_{t=1}^T \int_0^1  r_t(v) \cdot f^*_t(v) \, d\mu(v) \leq I, & \notag \\
& U_t(v) = \delta_t \cdot v \cdot r_t({v} ) - p_t({v}) + (1 -r_t({v})) \cdot U_{t+1}(v) &&\text{ for all }v\in[0,1], t \in [T], &\tag{IC} \label{eq:IC}\\
& p_t(0)=0 &&\text{ for all }t \in [T].& \tag{VP} \label{eq:VP}
\end{align}
Here, the first constraint family keeps track of the buyers that are present in the system at each point in time, as defined in \eqref{eq:fstar}. Moreover, in \eqref{eq:utility} the utility function is defined.
The second constraint makes sure that the seller only sells the available inventory $I$. This is true not only in expectation but for every outcome (since the number of buyers is infinite).
The third constraint family, makes sure that the optimal utility of a buyer is achieved when bidding truthfully and follows directly from \eqref{eq:utility}. The last constraint family normalizes the expected price to 0 at value 0, to ensure voluntary participation \eqref{eq:VP}. Note that, the incentive compatibility \eqref{eq:IC} and voluntary participation \eqref{eq:VP} constraints are only imposed ex-ante. This means that buyers' expected utility is maximum and non-negative when revealing their true valuation, but not necessarily for all possible outcomes. Interestingly, our main results establish that the optimal mechanism, even when optimizing over ex-ante \eqref{eq:IC} and \eqref{eq:VP} mechanisms, turns out to satisfy the constraints ex-post.

We now rewrite problem \eqref{eq:RM} eliminating the $f^*$ variables. Note that the recursion \eqref{eq:fstar} readily implies that $f^*_t(v) = \sum_{j=1}^t f_j(v)\cdot \prod_{j\leq i < t} (1-r_i(v))$. so that \eqref{eq:RM} becomes

\begin{align*}
     \max_{r_t, p_t, t\in[T]}\,  \sum_{t=1}^T \int_0^1 p_t(v) \cdot 
     \left(
     \sum_{j=1}^t f_j(v)\cdot \prod_{j\leq i < t} (1-r_i(v))
     \right)
     \, d\mu (v) 
\end{align*}
\vspace{-3ex}
\begin{align*}
 s.t. \, & (U), && \\
&\sum_{t=1}^T \int_0^1  \left(1- \prod_{j=t}^T(1-r_j(v)) \right)
\cdot f_t(v) \, d\mu(v) \leq I, & \\
& U_t(v) = \delta_t \cdot v \cdot r_t({v} ) - p_t({v}) + (1 -r_t({v})) \cdot U_{t+1}(v) &&\text{ for all }v\in[0,1], t \in [T], &(\text{IC})\\
& p_t(0)=0 &&\text{ for all }t \in [T].& (\text{VP})\\
\end{align*}
The inventory constraint follows directly by replacing the term $f^*(t)$. Intuitively this new formula is just a change of perspective. In the original formula, for each period $t$, we integrate the mass of buyers that get the item in period $t$, while here, we integrate the mass of buyers that arrive at $t$ and eventually get a copy of the item. 


\section{Optimal Mechanisms for Multi-Period Markets}

This section aims to prove \Cref{thm:rationing_optimal,thm:pricing_optimal}. 
We start by showing
\Cref{lem:derivative_utility,lem:utility_formula}, which characterize the utilities of the buyers in terms of the allocation functions $r$. Since the payments can be derived from the utilities, this also characterizes the payments in an incentive-compatible mechanism as a function of $r$. However, the converse is not true: if we apply the formula to arbitrary allocation functions $r$ to obtain payments $p$, the resulting mechanism is not necessarily incentive-compatible. 
In \Cref{lem:r_monotone}, we fill this gap. We show that if the functions $r$ are monotone, then the resulting mechanism is incentive-compatible. Conversely, we also show that for any mechanism, there are monotone functions $r$ that describe an incentive-compatible mechanism that results in at least as much revenue as the original one. In the proof of \Cref{lem:r_monotone}, we have to carefully take into account the discount factors. We defer the proof for the case of weakly and not strictly decreasing discount factors to \Cref{subsec:weakly_decreasing}.

To prove \Cref{thm:rationing_optimal,thm:pricing_optimal}, we exploit that by \Cref{lem:r_monotone}, we can restrict our attention to monotone functions $r$. Then, we can use the characterization of the payments as a function of $r$ to rewrite the revenue maximization problem of the seller as a problem with a multi-linear objective function. \Cref{lem:kplus1steps} allows us to characterize optimal solutions to this problem.


\begin{lemma}
    \label{lem:derivative_utility}
    For an anonymous and incentive-compatible mechanism $M$ with allocation functions $r_1,\dots,r_T$, we have that for all $t\in [T]$, $U_t$ is a convex function and
\begin{align}
\delta_t \cdot r_t(v) +(1 -r_t(v)) \cdot \partial U_{t+1}(v) \subseteq \partial  U_t(v). \label{eq:subdiff}
\end{align}
Moreover, for almost all $v\in [0,1]$, $\frac{d}{dv}U_t$ is well defined and satisfies
\begin{align}
    \frac{d}{dv} U_t(v) = \delta_t \cdot r_t(v) +(1 -r_t(v)) \cdot \frac{d}{dv} U_{t+1}(v).
    \label{eq:derivative_utility}
\end{align}\end{lemma}

\begin{proof}
Let $(p_t)_{t \in [T]}$ be the payment functions corresponding to $M$. Recall that as described in \eqref{eq:utility}, the utility of a buyer of type $v$ conditional on being present in period $t$ is given by
\[
U_t(v) = \max_{\bar{v} \in [0,1]} \left\{
\delta_t\cdot v \cdot r_t(\bar{v} ) - p_t(\bar{v}) + (1 -r_t(\bar{v})) \cdot U_{t+1}(v)
\right\}.
\]
By induction, starting from $U_{T+1}(v)=0$ and going backward, we have that $U_t$ is a convex function for every $t \in [T]$, since the maximum of a family of convex functions is also convex.
By using the envelope theorem and exploiting that the mechanism is truthful and thus the maximum is attained at $\bar{v}=v$, we get \Cref{eq:subdiff}.
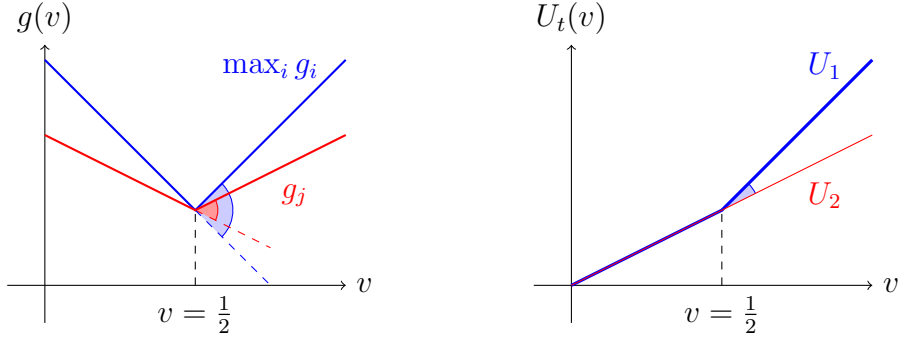
\begin{figure}[t]
\centering
    \usetikzlibrary{angles,quotes}
\begin{tikzpicture}[scale=1]

 \begin{scope}

  \draw[->] (-0.5,0) -- (4,0) node[right] {$v$};
  \draw[->] (0,-0.5) -- (0,3.2) node[above] {$g(v)$};

  \coordinate (A) at (0,0);
  \coordinate (B) at (2,1);
  \coordinate (C) at (4,3);
  \coordinate (D) at (4,2);
  \coordinate (E) at (0,3);
  \coordinate (F) at (0,2);
  \coordinate (G) at (3,0);
  \coordinate (H) at (3,0.5);

     \draw pic[draw=blue, fill=blue!20, angle eccentricity=1.2, angle radius=0.5cm]
    {angle=G--B--C};
    \draw pic[draw=red, fill=red!40, angle eccentricity=1.2, angle radius=0.3cm]
    {angle=H--B--D};
    \draw[dashed, red] (B) -- (H);
    \draw[dashed, blue] (B) -- (G);

  \draw[blue, thick] (E) -- (B) -- (C);
 \draw[red, thick] (F)-- (B) -- (D);

  \draw[dashed] (2,0) -- (2,1);

  \node[below] at (2,0) {$v=\frac{1}{2}$};

  \node[above, blue] at (3,2.6) {$\max_i g_i$};
  \node[right, red] at (3,1.2) {$g_j$};

   \end{scope}

    \begin{scope}[shift={(7,0)}]

  \draw[->] (-0.5,0) -- (4,0) node[right] {$v$};
  \draw[->] (0,-0.5) -- (0,3.2) node[above] {$U_t(v)$};

  \coordinate (A) at (0,0);
  \coordinate (B) at (2,1);
  \coordinate (C) at (4,3);
  \coordinate (D) at (4,2);
  
    \draw pic[draw=blue, fill=blue!20, angle eccentricity=1.2, angle radius=0.5cm]
    {angle=D--B--C};
    
  \draw[blue, very thick] (A) -- (B) -- (C);
 \draw[red] (A) -- (D);

  \draw[dashed] (2,0) -- (2,1);

  \node[below] at (2,0) {$v=\frac{1}{2}$};

  \node[above right, blue] at (3,2.6) {$U_1$};
  \node[right, red] at (3,1.2) {$U_2$};

   \end{scope}

\end{tikzpicture}
    \caption{The picture has the goal to visualize \Cref{eq:subdiff}. On the left we see a schematic picture, where the blue curve is the maximum over several curves $g_i$. At each point the subdifferential of each such $g_i$ is a subset of the subdifferential of the maximum. This can be seen at $v=\tfrac{1}{2}$, where the derivative is not defined and subdifferential of the blue curve (marked in blue) contains the subdifferential of the red curve (marked in red). On the right the same situation is depicted but for a concrete game. In period 1 a good is sold for a fixed price of $\frac{1}{2}$. In period 2, the same good is sold for a price of $0$ but only assigned with probability $\frac{1}{2}$. For $\delta_t=1$ for all $t \in [T]$, this results in the utility curves $U_1$ and $U_2$, where $U_1$ is not differentiable at $v=\frac{1}{2}$. The derivative of $U_2$ which is $\frac{1}{2}$ is in the subdifferential of $U_1$ (which is all tangents with slope in $[\frac{1}{2},1]$ and marked in blue). Since $r_1(\frac{1}{2})=0$, \Cref{eq:subdiff} reduces to $\frac{1}{2} \in [\frac{1}{2},1]$.}
    \label{fig:subgradient}
\end{figure}
Since $U_t$ and $U_{t+1}$ are convex, they are differentiable almost everywhere and absolutely continuous.
Thus, $\frac{d}{dv} U_t(v)$ and $\frac{d}{dv} U_{t+1}(v)$ are well defined almost everywhere.
Note that when $\frac{d}{dv} U_t(v)$ is well defined, $\partial U_t(v)= \{\frac{d}{dv} U_{t+1}(v) \}$. Therefore, replacing this in \Cref{eq:subdiff} we obtain \Cref{eq:derivative_utility}.
\end{proof}

\begin{lemma}
\label{lem:utility_formula}
    For an anonymous and incentive-compatible mechanism $M$ with allocation functions $r_1,\dots,r_T$, we have that
    \begin{align*}
        U_t(v) = \int_0^v \sum_{j=t}^T \delta_j\cdot  r_j(u)\cdot \prod_{t\leq i<j} \left(1-r_i(u)\right)\, du.
    \end{align*}
\end{lemma}

\begin{proof}
    Iterating \eqref{eq:derivative_utility}  gives the following almost everywhere
        \begin{align*}
        \frac{d}{dv} U_t(v) =\sum_{j=t}^T \delta_j\cdot  r_j(u)\cdot \prod_{t\leq i<j} \left(1-r_i(u)\right).
    \end{align*}
Moreover, by iterating \eqref{eq:IC} and \eqref{eq:VP} at $v=0$, we obtain that $U_t(0)=0$ for all $t \in [T]$. Since $U_t$ is convex, it is absolutely continuous, so $U_t(v)=\int_0^v \frac{d}{du} U_t(u)\, du$, which gives the formula in the statement of the lemma.
\end{proof}

\begin{lemma}
\label{lem:r_monotone}
	For every anonymous mechanism $M$ there exists a collection of monotone non-decreasing functions $\bar{r}_1,\dots,\bar{r}_T$ corresponding to an anonymous, and incentive-compatible mechanism $\bar{M}$ that achieves as much revenue as $M$. Conversely, given non-decreasing functions $r_1,\dots, r_T$, there is an anonymous and incentive-compatible mechanism $M$, for which these functions describe the allocation.
\end{lemma}	

In what follows, we give a proof of \Cref{lem:r_monotone} for the case of strictly decreasing discount factor $\delta_t$. The case of weakly decreasing $\delta_t$ requires extra arguments that make the proof rather tedious, which we relegate to \Cref{subsec:weakly_decreasing}.

\begin{proof}
Note that we can assume that $M$ is incentive compatible. If it is not, we can modify the mechanism such that the players submit their bids first to some algorithm that computes an optimal strategy for the original mechanism.

Recall that from \Cref{lem:derivative_utility}, we have
\begin{align*}
\delta_t \cdot r_t(v) +(1 -r_t(v)) \cdot \partial U_{t+1}(v) \subseteq \partial  U_t(v).
\end{align*}

Since $U_t$ is convex, not only $\ \frac {d}{dv} U_t(v)$ is well defined almost everywhere, but it is also monotone and has limits from the left and from the right at every point. Therefore, we have that for all $t\in [T]$, $v \in [0,1]$
\begin{align*}
    \partial U_t(v) = \left[\lim_{u\rightarrow v^-} \frac{d}{du} U_t(u), \lim_{u\rightarrow v^+} \frac{d}{du} U_t (u) \right].
\end{align*}

Then, for the points in which the derivative is not well defined, \Cref{eq:subdiff} is equivalent to
\begin{align*}
    \lim_{u\rightarrow v^-} \frac{d}{du} U_t(u)
    &\leq
    \delta_t\cdot r_t(v) + (1-r_t(v))\cdot \lim_{u\rightarrow v^-} \frac{d}{du} U_{t+1}(u)\\
    &\leq 
    \delta_t\cdot r_t(v) + (1-r_t(v))\cdot \lim_{u\rightarrow v+} \frac{d}{du} U_{t+1}(u)
    \leq
    \lim_{u\rightarrow v^+} \frac{d}{du} U_t(u).
\end{align*}

Rearranging, we obtain that
\begin{align}
\lim_{u\rightarrow v^-} \frac{\frac{d}{du} U_t(u) - \frac{d}{du}U_{t+1}(u)}{\delta_t -\frac{d}{du}U_{t+1}(u)} \leq r_t(v) \leq \lim_{u\rightarrow v^+} \frac{\frac{d}{du} U_t(u) - \frac{d}{du}U_{t+1}(u)}{\delta_t -\frac{d}{du}U_{t+1}(u)}. \label{eq:r_monotone_quotient}
\end{align}

Now we show that, if the denominator never becomes $0$, the quotient inside the limits is a monotone non-decreasing function. 
Since the limits exist and coincide almost everywhere, that implies that $r_t$ is also a monotone non-decreasing function.
In fact, notice that the numerator is non-decreasing because it is the derivative of the convex function 
\[
 U_t(v) - U_{t+1}(v)  = \max_{\bar{v} \in [0,1]} \left\{
\delta_t\cdot v \cdot r_t(\bar{v} ) - p_t(\bar{v}) -r_t(\bar{v}) \cdot U_{t+1}(v) \right\} . 
\]
This reformulation results by subtracting $U_{t+1}(v)$ on both sides of \eqref{eq:utility}. The expression on the right is indeed convex as it is the maximum of convex functions. Thus, the numerator, which is the derivative of a convex function, is non-decreasing. The denominator is non-increasing because $\frac{d}{du} U_{t+1}(u)$ is non-decreasing. 

Finally, it is easy to show by induction that $\frac{d}{du} U_{t+1}(u)\leq \delta_{t+1} \leq \delta_t$. Thus, for a strictly decreasing discount factor,  we conclude that the denominator is never zero, so the quotient is well-defined almost everywhere, and it is non-decreasing. We simply choose $\bar r_t(v) =r_t(v)$ for all $t \in [T]$ and $v \in [0,1]$.

In the case that the discount factors are not monotonically strictly decreasing, the denominators in \Cref{eq:r_monotone_quotient} can become zero, so the previous argument is not sufficient to conclude the monotonicity of $r_t$. However, when the denominators become zero we can modify $r_t$ to obtain a mechanism $\bar{M}$ with monotone functions $\bar{r}_t$ that gets as much revenue as $M$. We give the full construction in
\Cref{subsec:weakly_decreasing}. 

To prove the last part of the lemma, we show now that for any non-decreasing functions $r_1,\dots,r_T$, there is an incentive-compatible mechanism for which these functions are the allocation functions.\footnote{This part does not require the discount factor to be strictly decreasing.} Notice that $U_{t+1}$ only depends on how the mechanism is defined in periods $t+1,\dots, T$. We define the mechanism recursively from $T$ backward. In period $t$, a buyer that bids $v$ gets the item with probability $r_t(v)$ and pays
\begin{align*}
    p_t(v) = \delta_t\cdot v \cdot r_t(v) +(1-r_t(v))\cdot U_{t+1}(v) - \int_0^v \delta_t \cdot r_t(u) + (1-r_t(u))\cdot \frac{d}{du}U_{t+1}(u) \, du.
\end{align*}
First, observe that we have that $p_t(0)=0$ and thus voluntary participation, since the first and the last summand vanish directly and moreover, we observed already that $U_{\bar t}(0)=0$ for all $\bar t \in [T]$.
Next, we show that with this payment function, in period $t$, it is optimal for a buyer of type $v$ to bid $\bar v = v$. In fact, for a buyer of type $v$, the utility of bidding $\bar v$ in period $t$ is
\begin{align}
    &\delta_t\cdot v \cdot r_t(\bar v) + (1-r_t(\bar v))\cdot U_{t+1}(v) \notag\\
    &\qquad-\left(\delta_t\cdot \bar v \cdot r_t(\bar v) +(1-r_t(\bar v))\cdot U_{t+1}(\bar v) - \int_0^{\bar v}  \delta_t \cdot r_t(u) + (1-r_t(u))\cdot \frac{d}{du}U_{t+1}(u) \, du\right) \notag \\
    &= \delta_t\cdot r_t(\bar v)\cdot (v-\bar v)
    + (1-r_t(\bar v))\cdot (U_{t+1}(v)-U_{t+1}(\bar v)) \notag \\
    &\qquad +\int_0^{\bar v}  \delta_t \cdot r_t(u) + (1-r_t(u))\cdot \frac{d}{du}U_{t+1}(u) \, du.
    \label{eq:utility_induction_payments_monotone_r}
\end{align}
Since $\frac{d}{du} U_{t+1}(u)\leq \delta_{t+1}\leq \delta_t$, and because $r_t(u)$ is non-decreasing, we have that for all $\bar v \in [0,1]$,
\begin{align*}
   & \int_{\bar v}^{v}  \delta_t \cdot r_t(u) + (1-r_t(u))\cdot \frac{d}{du}U_{t+1}(u) \, du\\
   & \geq 
    \int_{\bar v}^v \delta_t \cdot r_t(\bar{v}) + (1-r_t(\bar{v}))\cdot \frac{d}{du}U_{t+1}(u) \, du\\
   & \geq
    \delta_t\cdot r_t(\bar v)\cdot (v-\bar v)
    + (1-r_t(\bar v))\cdot (U_{t+1}(v)-U_{t+1}(\bar v)).
\end{align*}
Replacing this back in \Cref{eq:utility_induction_payments_monotone_r}, we obtain that it is maximized at $\bar v = v$.
\end{proof}

We use the following technical result, which is a consequence of Caratheodory's theorem and the sparsity of optimal solutions of linear programs. Equivalent observations have been used in mechanism design~\cite{K21,DS23} and stochastic optimization~\cite{N15}. For completeness, we present a proof in the following.

\begin{lemma}
\label{lem:kplus1steps}
Consider the problem of choosing a function $h:[0,1]\rightarrow [0,1]$ that maximizes a linear, continuous, and bounded operator $J(h)$, under the constraints that $h$ is non-decreasing and $G_i(h)\leq b_i$, for $i \in [k]$ and some linear, continuous, and bounded operators $G_1,\dots, G_k$ and constants $b_1,\dots, b_k$. There is an optimal solution $h^*$ that is a monotone step function with at most $k+1$ steps. Moreover, if it has exactly $k+1$ steps, then $h^*(1)=1$.
\end{lemma}

\begin{proof}
    For some finite measure $\mu$ of $[0,1]$, we assume the continuity of $J$ and $G_1,\dots,G_k$ as functions from $L^1([0,1],\mu)$ into $\R$.
Let $H$ be the set of single-step functions, i.e., of all functions of the form $h(x)= \mathds{1}_{\{a\leq x\}}$ or $h(x)= \mathds{1}_{\{a<x\}}$ for some $a\in [0,1]$. Notice that $H$ is a closed set in $L^1([0,1],\mu)$. Notice also that the set of monotone functions is the closure of the convex hull of $H$. In other words, any monotone function is the limit of a sequence of convex combinations of single-step functions. Denote this by $\cl(\conv(H))$. Consider the linear and continuous operator $\mathcal{J}$ from $L^1([0,1],\mu)$ to $\R^{k+1}$ defined as $\mathcal{J}(h)=(J(h),G_1(h),\dots,G_k(h))$. Since $\mathcal{J}$ is linear and continuous, and $H$ is a closed set, and because in $\R^{k+1}$ the convex hull of a bounded and closed set is always closed, we have that $\mathcal{J}(\cl(\conv(H))) = \cl(\conv(\mathcal{J}(H))) = \conv(\mathcal{J}(H))$. Consider an optimal solution $h^*\in \cl(\conv(H))$. By Caratheodory's theorem, $\mathcal{J}(h^*) \in \conv(\mathcal{J}(H))\subseteq \R^{k+1}$ can be written as a convex combination of $k+2$ points in $\mathcal{J}(H)$, which means there are functions $h_1,\dots, h_{k+2} \in H$ and coefficients $\alpha_1,\dots,\alpha_{k+2}$ such that $\mathcal{J}(h^*)=\sum_{i=1}^{k+2} \alpha_i \mathcal{J}(h_i)=\mathcal{J}\left( \sum_{i=1}^{k+2} \alpha_i h_i\right)$. Since $\mathcal{J}$ encodes feasibility and the value in the optimization problem, this means that $\hat{h}=\sum_{i=1}^{k+2} \alpha_i h_i$, which is a monotone step function with $k+2$ steps, is also an optimal solution. We can further reduce the number of steps to $k+1$ by noting that the linear program
\begin{align*}
    \max &\sum_{i=1}^{k+2} \alpha_i J(h_i), \\ 
    \text{ s.t. } & \sum_{i=1}^{k+2} \alpha_i G_j(h_i) \leq b_j, && \text{for all }j=1,\dots, k,\\ 
    & \sum_{i=1}^{k+2} \alpha_i =1,\\ 
    & \alpha_i \geq 0, && \text{for all }i=1,\dots, k+2.
\end{align*}
has $k+1$ constraints, so an optimal basic solution has at most $k+1$ non-zero variables.
Moreover, when there are $k+1$ non-zero variables resulting in $k+1$ steps, we have that all non-trivial constraints are tight and thus especially $\sum_{i=1}^{k+2} \alpha_i =1$, which results in $h^*(1)=1$.
\end{proof}

Now we are well-prepared to prove \Cref{thm:rationing_optimal,thm:pricing_optimal}.

\begin{proof}[Proof of \Cref{thm:rationing_optimal} and \Cref{thm:pricing_optimal}]

Because of \Cref{lem:r_monotone}, we can restrict our attention to monotone functions $r_1,\dots,r_T$ corresponding to an incentive-compatible mechanism $M$.
Given these functions and \Cref{lem:utility_formula}, we have a formula for the payments in terms of $r_1,\dots,r_T$, and integrating we get also a formula for the revenue.

Putting together \Cref{eq:utility}, the incentive compatibility, and \Cref{lem:utility_formula}, we have that
\begin{align*}
    p_t(v) &= \delta_t\cdot v\cdot r_t(v) + (1-r_t(v))\cdot U_{t+1}(v) - U_t(v)\\
    &= \delta_t\cdot v\cdot r_t(v) + (1-r_t(v))\cdot 
    \int_0^v \sum_{j=t+1}^T \delta_j\cdot  r_j(u)\cdot \prod_{t+1 \leq i<j} \left(1-r_i(u)\right)\, du\\
    &\qquad \qquad-
    \int_0^v \sum_{j=t}^T \delta_j\cdot  r_j(u)\cdot \prod_{t\leq i<j} \left(1-r_i(u)\right)\, du.
\end{align*}
Notice that from this formula we obtain that $p_t(v)$ is linear in each $r_j$ (separately, not jointly for all $j$), and that it depends only on $r_j$ with $j\geq t$.
Recall that the objective in problem \eqref{eq:RM} can be written as 
\begin{align*}
    \rev(M) &= \sum_{t=1}^T \int_0^1 p_t(v) \cdot 
     \left(
     \sum_{j=1}^t f_j(v)\cdot \prod_{j\leq i < t} (1-r_i(v))
     \right)
     \, d\mu (v) .    
\end{align*}

Note that for every $t \in [T]$, $i \in [T]$, we have that in the $t$-th term of the sum, $r_i$ appears only once. Therefore, the revenue is linear in each $r_i$.


For a given period $t$, consider the problem of maximizing the revenue as a function of $r_t$, taking all other functions as fixed. We already showed that the revenue is linear in $r_t$, and the only constraints we have are that $r_t$ must be non-decreasing and that the mechanism does not sell more items than the inventory $I$. 
Recall from \eqref{eq:RM} that the inventory constraint can be written as:
\begin{align*}
    \sum_{t=1}^T \int_0^1 \left(1-\prod_{j=t}^T (1-r_j(v)) \right) \cdot f_t(v) \, d\mu(v) \leq I,
\end{align*}
which is also linear in $r_t$.
Applying \Cref{lem:kplus1steps} for $k=1$, we obtain that there is an optimal solution $r^*_t$ that is a step function with at most $2$ steps. Moreover, if the inventory constraint is non-binding, then we can apply the lemma for $k=0$, so the optimal solution is a single-step function.

For the case of a single-step function, we show that the allocation can be implemented with a posted price. Let $q_t$ be the value at which the step occurs, i.e., $r^*_t(v)$ is zero for $v<q_t$ and one for $v>q_t$. We take the price
\begin{align*}
    p_t=\delta_t\cdot q_t - U_{t+1}(q_t).
\end{align*}
For a buyer of type $v$, the utility that results from buying at this price is
\begin{align*}
    \delta_t\cdot v -p_t = \delta_t \cdot (v-q_t) + U_{t+1}(q_t),
\end{align*}
whereas the utility of waiting for the next period is $U_{t+1}(v)$. Notice that these two are equal at $v=q_t$, so a buyer of type $v=q_t$ is indifferent. Notice also that the utility from buying at $t$ has derivative $\delta_t$, while  $\frac{d}{dv}U_{t+1}(v)\leq \delta_{t+1}$. 
Thus in case $\delta_t> \delta_{t+1}$ buyers of type $v<q_t$ strictly prefer waiting and buyers of type $v>q_t$ strictly prefer buying at $t$. If $\delta_t=\delta_{t+1}$, both derivatives could coincide, and in that case, buyers of type $v>q_t$ are also indifferent.\footnote{In this case,  $r_t^*$ is an equilibrium, but there could be other equilibria beside $r_t^*$.} Therefore, posting the price $p_t$ implements the allocation $r_t^*$.

For the case of a two-step function, we show that the allocation can be implemented with two prices: a high price at which the buyer gets the good for sure and a limited inventory at a lower price. Let $q_t^\ell < q_t^h$ be the points at which the steps of $r_t^*$ occur. Let $r_t^\ell= \lim_{v\rightarrow q_t^\ell+} r_t^*(v)$. We set prices
\begin{align*}
    p_t^\ell &= r_t^\ell\cdot \left(\delta_t\cdot q_t^\ell-U_{t+1}(q_t^\ell)\right),\\
    p_t^h &= \delta_t\cdot q_t^h - \left(
    r_t^\ell \cdot \delta_t\cdot q_t^h - p_t^\ell + (1-r_t^\ell)\cdot U_{t+1}(q_t^h)
    \right).
\end{align*}
We set the inventory to be sold at the low price so that the probability of getting the good is exactly $r_t^\ell$, that is, the mass of buyers that want to pay the low price divided by $r_t^\ell$.
Note that $f^*$ (see \eqref{eq:fstar}) is well-defined since $r_t$ is given. 
\begin{align*}
    I_t^\ell = \frac{1}{r_t^\ell} \cdot \int_0^1 \mathds{1}_{\{r_t^*(v)=r_t^\ell\}} \cdot f^*_t(v) \, d\mu(v)
\end{align*}
The prices are set so that a buyer of type $v=q_t^\ell$ is indifferent between waiting and paying the low price, and a buyer of type $v=q_t^h$ is indifferent between the low price and the high price. With a similar analysis as before, we can check that other buyers strictly prefer one of the options, so we implement $r_t^*$.
Note that this case might be degenerate in the sense that $r_t^*$ is actually a one-step function. If $r_t^*(1)=1$, then the jump point is treated as $q_t^h$. This means we define a posted price as described in the first case. If $r_t^*(1)<1$, then the jump point is handled as $q_t^\ell$. We choose the prices and the inventory as described in the case of two steps with the only difference that $q_t^h$ is chosen high enough such that it does not play any role. 
\end{proof}

 \section{Proof of \Cref{lem:r_monotone} for Weakly Decreasing Discount Factors}
\label{subsec:weakly_decreasing}

In the previous section, we proved \Cref{lem:r_monotone} in the case of strictly decreasing discount factors. If the discount factors are weakly decreasing, then the proof of the monotonicity of $r_t$ fails.
Recall that in the previous section, we proved that
\begin{align*}
\lim_{u\rightarrow v^-} \frac{\frac{d}{du} U_t(u) - \frac{d}{du}U_{t+1}(u)}{\delta_t -\frac{d}{du}U_{t+1}(u)} \leq r_t(v) \leq \lim_{u\rightarrow v^+} \frac{\frac{d}{du} U_t(u) - \frac{d}{du}U_{t+1}(u)}{\delta_t -\frac{d}{du}U_{t+1}(u)}. 
\end{align*}
In particular, we showed that $\frac{d}{du} U_{t+1}(u)\leq \delta_{t+1}$, so if the discount factor is strictly decreasing, we can show that the denominator of the above quotient is nonzero and thus, it is well-defined. Moreover, we showed that in this case the quotient is monotonically increasing. However, this argument does not work if $\delta_t=\delta_{t+1}$ for some $t$.

In fact, if $\delta_t=\delta_{t+1}$, $r_t$ can decrease in an interval where $\delta_t -\frac{d}{du}U_{t+1}(u)=0$. We will show how to modify $r$ in such intervals to obtain new monotone functions $\bar r$ without reducing the revenue. The rest of the proof of \Cref{lem:r_monotone} does not rely on the strict monotonicity of the discount factors. In particular, monotone allocation functions are sufficient to define an incentive-compatible mechanism. Therefore, our new functions $\bar r$ define an incentive-compatible mechanism.


So, consider a sequence of weakly decreasing discount factors. 
Let $\bar T$ be a maximal interval of indices such that $\delta_t$ is the same for all $t \in \bar T$. 
We define $I$ to be the set of all pairs $(t,v)$ such that $t\in \bar{T}$ and $\delta_t -\frac{d}{dv}U_{t}(v)=0$. We start by describing the shape of $I$, which is stair-like as depicted in \Cref{fig:I}.

\begin{figure}
    \centering
    \includegraphics[width=0.65\textwidth]{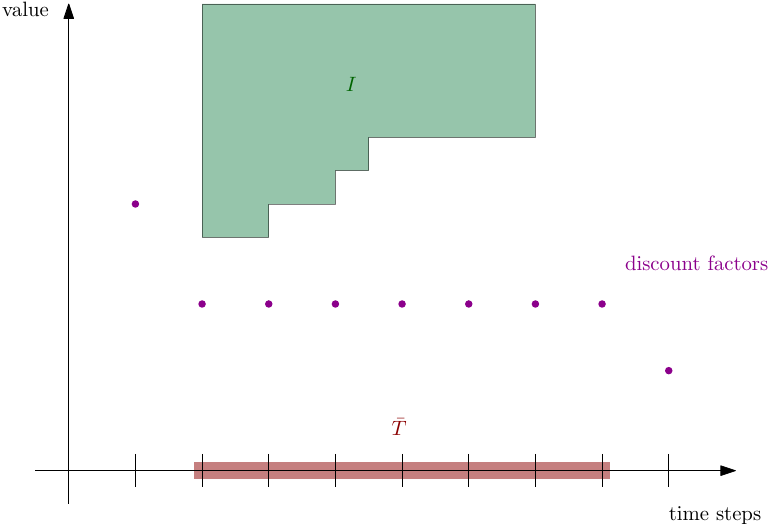}
    \caption{The marked interval of time slots is $\bar T$, the dots depict the discount factors and the region marks the shape of region $I$.}
    \label{fig:I}
\end{figure}

\begin{claim}
\label{claim:staircase}
    If $(t,v) \in I$, then for all $v'\geq v$ and all $t'\in \bar{T}$ such that $t'\leq t$, we have $(t',v') \in I$.
\end{claim} 
\begin{proof}
Since $\frac{d}{dv} U_t(v) \leq \delta_t$ for all $v\in [0,1]$ and all $t=1,\dots,T$, and because of the monotonicity of $\frac{d}{dv} U_t$, we have that if $\frac{d}{dv} U_t(v) = \delta_t$ then $\frac{d}{dv} U_t(v') = \delta_t$ for all $v'\geq v$. We also have that
\[
\frac{d}{dv} U_t(v) = \delta_t\cdot r_t(v) + (1-r_t(v))\cdot \frac{d}{dv} U_{t+1}(v),
\]
so, if $t,t+1\in \bar{T}$ and $(t+1,v)\in I$, then $\frac{d}{dv}U_{t+1}(v)=\delta_{t+1}=\delta_t$. Replacing this in the previous equation, we obtain that $ \frac{d}{dv}U_t(v)=\delta_t$, and therefore, $(t,v)\in I$.    
\end{proof}
Intuitively, a point $(t,v)$ in $I$ corresponds to a buyer that gets the item for sure within $\bar{T}$. We formalize this intuition in the following claim.

\begin{claim}
\label{claim:corners_one}
    For every $(t,v)$ in $I$, there is $(t',v)\in I$ with $t'\geq t$ such that $r_{t'}(v)=1$.
\end{claim}
\begin{proof}
    Notice that, by \Cref{claim:staircase}, it is enough to show that if $(t,v)\in I$ but $(t+1,v)\notin I$, then $r_t(v)=1$. So, consider a pair $(t,v)\in I$ such that $(t+1,v)\notin I$, and assume by contradiction that $r_t(v)<1$. We have that
    \begin{align*}
    \delta_t =\frac{d}{dv}U_{t}(v) = \delta_{t} \cdot r_{t} (u) + (1-r_{t}(v)) \cdot \frac{d}{dv}U_{t+1}(v)
    \leq \delta_{t} \cdot r_{t} (v) + (1-r_{t}(v)) \cdot \delta_{t+1} \leq \delta_t.
\end{align*}
    Since we have $\delta_t$ on both ends, all inequalities hold with equality. However, since $(t+1,v)\notin I$, either $\frac{d}{dv}U_{t+1}(v)<\delta_{t+1}$, or $\delta_{t+1}<\delta_t$. In both cases we get a contradiction.
\end{proof}

Now we can define the modified allocation function, that assigns an item within $I$ immediately. Formally, we define  
\[
\bar r_t(v) = \begin{cases} 1 \quad &\text{ if } (t,v) \in I  \\
r_t(v) &\text{ else.}
\end{cases}
\]

The function $\bar r$ is monotone by  \Cref{claim:staircase} and monotonicity of $r$ if $(t,v) \not\in I$. 

Let $\bar M$ be the mechanism corresponding to $\bar r$. The next Claim shows that the surplus of any particular buyer is equal under $M$ or $\bar M$.

\begin{claim}
\label{claim:r_and_rbar}
It holds that
   \begin{equation}
\label{eq:r_and_rbar_same}
    \sum_{k=t}^T \delta_k r_k(v) \prod_{t \leq i < k } (1 -r_i(v)) = \sum_{k=t}^T \delta_k \bar r_k(v) \prod_{t \leq i < k } (1 - \bar r_i(v)).
\end{equation} 
\end{claim}

\begin{proof}
    Fix $(t,v)$. If $(t,v)$ is such that there is no $(t',v') \in I$ with $t' \geq t$, then $r$ and $\bar r$ coincide and the equation holds. 

    Otherwise let $t_{max}$ be the maximal $t\in T$ such that $(t_{max}, v) \in I$ and $t_{max}\geq t$. Moreover, let $t_{min}$ be the maximum of the minimal such value and $t$.

    Since $\bar r_{t'}(v')=1$ for all $(t',v')\in I$,
    \[
    \sum_{k=t}^T \delta_k \bar r_k(v) \prod_{t \leq i < k } (1 - \bar r_i(v)) = \sum_{k=t}^{t_{min}-1} \delta_k r_k(v) \prod_{t \leq i < k } (1 -  r_i(v)) + \delta_{t_{min}}\prod_{t \leq i < t_{min} } (1 -  r_i(v)).
    \]

    For $r$ we first use that $\delta_{t'}$ is the same for all times $t'$ with $t_{min}\leq t' \leq t_{max}$. Then, we use that $r_{t_{max}}(v,t)=1$ by \Cref{claim:corners_one}. This allows to exploit a telescope behavior, in the sense that the second sum boils down two $r_{t_{max}-1}(v,t) x + (1- r_{t_{max}-1}(v,t))x$ and then we can push this argument further until the whole sum reduced to a single term. 
    
    \begin{align*}
    \sum_{k=t}^T \delta_k r_k(v) \prod_{t \leq i < k } (1 - r_i(v)) &= \sum_{k=t}^{t_{min}-1} \delta_k r_k(v) \prod_{t \leq i < k } (1 -  r_i(v)) + \sum_{k=t_{min}}^{t_{max}} \delta_k r_k(v) \prod_{t \leq i < k } (1 -  r_i(v))\\
    &= \sum_{k=t}^{t_{min}-1} \delta_k r_k(v) \prod_{t \leq i < k } (1 -  r_i(v)) + \delta_{t_{min}} \sum_{k=t_{min}}^{t_{max}} r_k(v) \prod_{t \leq i < k } (1 -  r_i(v))\\
    &= \sum_{k=t}^{t_{min}-1} \delta_k r_k(v) \prod_{t \leq i < k } (1 -  r_i(v)) + \delta_{t_{min}} \prod_{t \leq i < t_{min} } (1 -  r_i(v))
    \end{align*}
\end{proof}

\begin{claim}
    The revenue of mechanism $\bar M$, defined by $\bar r$, equals that of mechanism $M$.
\end{claim}

\begin{proof}
To show that the revenue remains the same we aim to express the revenue as the social welfare minus the sum of the utilities. We show that both these parts, the social welfare and the total utilities are the same in $M$ and $\bar M$.

\paragraph{Social welfare.}
We obtain the following description to compute the social welfare for mechanism $M$ with allocation functions $r$. 
\begin{align*}
    \sum_{t=1}^T \delta_t \int_0^1 f^*_t(v) r_t(v) \, d\mu(v)
     = & \sum_{t=1}^T \int_0^1 \delta_t  r_t(v) \cdot \left( \sum_{j=1}^t f_j(v) \prod_{j \leq i < t } (1 -r_i(v)) \right)  \, d\mu(v) \\
      = & \sum_{t =1}^T \int_0^1  f_t(v) \left( \sum_{k=t}^T \delta_k r_k(v) \prod_{t \leq i < k } (1 -r_i(v))
        \right)
  \, d\mu(v).
\end{align*}

 Now, by \Cref{claim:r_and_rbar}, the coefficient of $f_t$ is the same if we replace $r$ with $\bar r$, and therefore, the social welfare is the same in $M$ and $\bar M$.

\paragraph{Utilities.}
We also show that the utilities are the same for $M$ and $\bar M$.

Consider
\begin{align*}
 \sum_{t=1}^T \int_0^1 U_t(v) f_t(v) \, d\mu(v)
  = & \sum_{t=1}^T \int_0^1 \int_0^v \sum_{j=t}^T \delta_j r_j(u) \prod_{t \leq i < j } (1-r_i(u)) du  f_t(v) \, d\mu(v) \\
  = & \sum_{t=1}^T \int_0^1 f_t(v) \left(\int_0^v \sum_{j=t}^T \delta_j r_j(u) \prod_{t \leq i <j } (1-r_i(u)) du 
    \right) 
  \, d\mu(v).
\end{align*}
 Again, by using \Cref{claim:r_and_rbar}, we observe that for each value $u$ the terms in the integral coincide for $r$ and $\bar r$. Thus, also the integrals and, thus, the social welfare coincide.

Finally, we obtain that also the revenue in the mechanisms $M$ and $\bar M$ coincide. This follows directly from the fact that the revenue can be expressed as the social welfare minus the total utility. Since both the social welfare and the total utility coincide for $M$ and $\bar M$, the same follows for the revenue.
\end{proof}

We just showed that we can modify $r$ on a subset of time steps where all $\delta$-values coincide to obtain a monotone function. But by repeating this argument for all relevant subsets of time steps we end up with the final monotone function $\bar r$.

\section{General Discounting}
\label{sec:discount}

A natural generalization of our model is to add  
discount factors on future monetary transfers both for the seller as well as for the buyers. We discuss that our main results and proofs transfer to this setting, which demonstrates the generality of our approach. 

Let $\lambda^S_t \in (0,1]$ for all $t \in [T]$ be an arbitrarily decreasing sequence of discount factors on the seller's profit, i.e., the seller scales the payments she receives at time $t$ by $\lambda^S_t$ and we have $\lambda^S_{1}\geq \ldots \geq \lambda_{T}^S $ .
To be precise, the revenue of the seller in this setting is
\begin{equation}
\label{eq:revenue_discounted}
 \rev(M)=\sum_{t=1}^T \lambda^S_t \int_0^1 p_t(v) \cdot f^*_t(v)\, d\mu (v).    \tag{rev disc}
\end{equation}
This coincides with \eqref{eq:revenue_general} for $\lambda^S_t=1$ for all $t \in [T]$.

Let $\lambda^B_t \in (0,1]$ for all $t \in [T]$ be an arbitrary decreasing sequence of discount factors on the buyers' money. That is, a buyer of type $v$ that gets the item in period $t$ and pays $p$, gets a utility of $\delta_t\cdot v- \lambda^B_t\cdot p$.

With these discount factors, some care is needed when we define the class of feasible mechanisms. To see this, consider the following example. Assume the seller and the buyer have a different discount sequences, e.g., let $\alpha > \beta$ and define 
$\lambda^B_t \coloneqq \alpha^t$, and  $ \lambda^S_t \coloneqq \beta^t$ for all $t \in [T]$.
Then, the seller and the buyer can both achieve a positive (and even unbounded revenue) without even exchanging a good, as follows. 
At time $0$, the buyer makes a payment of $P$ to the seller, and at time 1, the seller pays $P/(2 \alpha) + P/(2\beta)$ to the buyer. After this exchange, the revenue of the seller is 
$
P- \left(P/ (2\alpha) + P/(2\beta)\right)\cdot \beta = (P/2)\cdot  (1- \beta/\alpha )>0.
$
And the revenue of the buyer is 
$
-P + (P/ (2\alpha) + P/(2\beta)) \cdot \alpha = (P/2)\cdot ( \alpha/\beta -1) >0.$
Since both the buyer and the seller profit from this exchange and their profit is unbounded if $P$ goes to infinity they could achieve unlimited utility.

We restrict our attention to mechanisms that rule out such behavior. To do so, we make the following two assumptions:
\begin{enumerate}[(i)]
    \item prices are nonnegative, \label{it:non_neg}
    \item every payment from the buyer to the seller is made exactly at the point in time when the buyer gets an item. \label{it:no_time_delay}
\end{enumerate}
On an intuitive level, assumption \eqref{it:non_neg} means there are no payments from the seller to the buyers, so borrowing money, as in the example, is forbidden. Mathematically, we need this assumption to obtain the formula in \Cref{lem:utility_formula}. In fact, to prove the formula, we need that the utility equals the integral of its derivative, for which we need that a buyer with valuation 0 has utility 0. It is easy to check that if payments are nonnegative, $U_T(0)=0$, and by induction, we have the same for every period. Note that in the example this is not satisfied. Even buyers of type $0$ get positive utility with the borrowing scheme.

Assumption \eqref{it:no_time_delay} means buyers are not allowed to pay via a credit (in which case they would get the good today and pay the money tomorrow), and buyers cannot pay in advance. This assumption is needed to rule out that profit is achieved by timing the payment well.  
Payments in advance are already ruled out by the anonymity of the mechanism. The mechanism is supposed to forget everything that happened in preceding periods and decide based on the current period and the bids. Otherwise, the mechanism could not be expressed by the functions $r_t$ and $p_t$.
Mathematically, we need that buyers do not pay later to be consistent with the formula $$U_t(v) = \max_{\bar{v} \in [0,1]} \left\{
\delta_t\cdot v \cdot r_t(\bar{v} ) - \lambda^B_t\cdot p_t(\bar{v}) + (1 -r_t(\bar{v})) \cdot U_{t+1}(v)
\right\}.$$ Note that, in this formulation we assume that the utility of period $t+1$ is only relevant with probability $(1-r_t(v))$, i.e., only when the buyer does not receive the good.
Thus, assumptions \eqref{it:non_neg} and \eqref{it:no_time_delay} are necessary and well-motivated.

To see that the proofs of \Cref{thm:rationing_optimal,thm:pricing_optimal} still work under this extension, first observe that under conditions (i) and (ii), the proof of \Cref{lem:kplus1steps} is not affected at all. 

Conceptually, it is easy to see that the proof of \Cref{lem:derivative_utility} transfers, as it relies on applying the envelope theorem to the recursive formula for the utility. Notice that the payments do not appear in the derivative, so the discount factor does not affect the proof. Then to prove \Cref{lem:utility_formula} the formula developed in \Cref{lem:derivative_utility} is iterated, thus the proof remains the same.

The proof of \Cref{lem:r_monotone} consists of two parts. First, we show that the functions $r$ corresponding to a given mechanism are either already monotone in $v$ or can be redefined in an appropriate way to fulfill this. This first part is based on applying the formula of the utility developed in \Cref{lem:utility_formula} and is thus not affected. 
In the second part of the proof, given some collection of monotone functions $r$ we define price functions. Here, apparently the discounting factor leads to a scaling of the price vector. 
Namely, 
\begin{align*}
    p_t(v) = \frac{1}{\lambda_t^B} \left(\delta_t\cdot v \cdot r_t(v) +(1-r_t(v))\cdot U_{t+1}(v) - \int_0^v \delta_t \cdot r_t(u) + (1-r_t(u))\cdot \frac{d}{du}U_{t+1}(u) \, du.\right)
\end{align*}
This factor cancels out directly, when analyzing the utility such that the remaining proof transfers without any changes.

The proofs of \Cref{thm:rationing_optimal,thm:pricing_optimal} transfer directly, since the formula of the revenue remains linear in $r_t$ for $t \in [T]$ if we add the discount factor. More precisely, the formula of $p_t(v)$ gets scaled by a factor of $\frac{1}{\lambda_t^B}$, while the formula of $f^*$ does not change at all. Thus, as without the new discount factors, $p_t$ depends linearly on $r_j$ with $j \geq t$ and $f^*_t$ depends linearly on $r_j$ for $j <t$. Thus the revenue
\begin{align*}
    \rev(M) &= \sum_{t=1}^T \lambda^S_t \int_0^1 p_t(v)\cdot f^*_t(v) \, d\mu(v).
\end{align*}
remains linear in each $r_j$ (one by one) despite the discount factor.
The only constraint is the inventory constraint, which just depends on $f_t^*$ and thus remains linear as well.

This discussion shows that our main results easily extend to a more general model where the buyer and the seller both have a discount rate and these rates are not correlated.

\section{Further Research Directions}

In this paper, we presented simple, revenue optimal mechanisms for large multi-period markets.
In the revenue optimal mechanism, the seller chooses a price and a capacity per period. If we further assume supply is unbounded, capacities are no longer necessary, meaning the optimal mechanism is just a posted price per period. 

Motivated by the wish for simple mechanisms, there are two natural follow-up questions. Firstly, what is the comparative revenue achieved by posted prices versus any anonymous mechanism when supply is limited? Put differently, what is the trade-off for simplicity when solely considering posted pricing mechanisms? 
Secondly, how many lotteries has a seller to run in an optimal mechanism? In other words, in how many periods the seller must set the capacity strictly between $0$ and $\infty$. In \Cref{ex:rationing_usefull_inventory} we saw that at least one lottery is necessary. Does this suffice? An important reason to believe it should not be necessary to use rationing in every period is that there is a single supply constraint. If the revenue were linear in all allocation functions, LP sparsity would imply that we need rationing only once.

These two questions are related since a positive answer to the latter (one lottery is sufficient), is likely to lead to a small quotient between the optimal revenue of any anonymous mechanism and the optimal revenue of posted prices. 

Furthermore, the same question can also be asked in an atomic setting where buyers are not infinitesimally small. Here, even in the case of unlimited inventory, the use of lotteries helps to increase the objective function value, as seen in \Cref{ex:rationing_usefull_discrete}. The example shows that the seller can achieve a revenue of $\tfrac{5}{4}$ when using lotteries and of at most $1$ when using only prices.
In the relatively simple setting with two buyers, it is possible to explicitly find the instance in which the ratio between the seller's revenue using only prices and that when she can use prices and lotteries is maximum. It turns out that the largest possible ratio is indeed $\tfrac{5}{4}$, and it is attained by the example above. 

When moving to a setting with more buyers and (relevant) periods, the analysis becomes more involved. Consider a three-period setup with three buyers arriving at time one and with valuations $1$, $1/2$, and $1/3$. First, observe that a seller choosing only prices can achieve a revenue of $1$. However, if the seller can set capacities, the revenue can be increased to $31/24$. This revenue can be achieved with the strategy $(p_1,c_1)=(\frac{13}{24}, \infty)$,   $(p_2,c_2)=(\frac{5}{12}, 1)$,  $(p_3,c_3)=(\frac{1}{3}, 1)$, where $p_i$ is the price in period $i$ and $c_i$ is the amount sold in period $i$ for $i\in [3]$.
Although now the strategy space is larger, it can still be checked that it is optimal for buyer $i$ to buy at time $i$ for $i\in [3]$. The resulting revenue is     
$
p_1 +  p_2 + p_3 = \frac{13}{24} + \frac{5}{12} + \frac{1}{3} = \frac{31}{24}.
$
Proving that this is indeed the best possible strategy for the seller turns out to be quite tedious. Moreover, giving a general bound on the fraction of the optimal revenue a pricing mechanism can achieve is an exciting question for future research. This would require a precise formulation of the discrete model and its equilibria, which is less immediate than in the non-atomic case, where there is a clear, natural choice.

A different direction is to explore the existence of an algorithm to compute the optimal strategy for the seller, encompassing both optimal capacities and prices. In the setting with unbounded supply, the algorithm of \citet{B08} can compute the optimal prices. Can this result be generalized to compute the optimal prices and capacities?

\vspace{1cm}

\noindent
\textbf{Acknowledgments.}\quad  We would like to
thank Juan Escobar, Benny Moldovanu, and Andrzej Skrzypacz for fruitful discussions and feedback, that greatly improved the presentation of the paper.
We also thank Simon Board, Simon Loertscher, Balasubramanian Sivan, and Philipp Strack for their helpful comments.

\vspace{1cm}

 \bibliographystyle{ACM-Reference-Format}

 \bibliography{literature}

\end{document}